\documentclass[aps,pra,twocolumn,letterpaper,superscriptaddress,floatfix,nofootinbib]{revtex4-1}
\usepackage[ruled,vlined]{algorithm2e}
\usepackage[braket,qm]{qcircuit}
\usepackage{amsmath}
\usepackage{amssymb}
\usepackage{bm}
\usepackage{epsfig}
\usepackage{epstopdf}
\usepackage{subfigure}
\usepackage{array}
\usepackage{multirow}
\usepackage{eucal}
\usepackage{makecell}%bold the line
\usepackage{fnbreak}

\usepackage{amsthm}
\newtheorem{thm}{Theorem}%[section]
\newtheorem{lem}[thm]{Lemma}

\usepackage{amsmath,amssymb}
\usepackage{graphicx}
\usepackage{color}
\usepackage{url}
\usepackage{textcomp} % This package is just to give the text quote '
\usepackage{enumitem}
\setlist{noitemsep,leftmargin=*}

\usepackage{hyperref}
\hypersetup{
colorlinks=true,
urlcolor= blue,
citecolor=blue,
linkcolor= blue,
}
\begin{document}

\title{Quantum Data Fitting Algorithm for Non-sparse Matrices}

\author{Guangxi Li} 
\email[]{Guangxi.Li@student.uts.edu.au}
\affiliation{Centre for Quantum Software and Information, Faculty of Engineering and Information Technology, University of Technology Sydney, Sydeny, NSW 2007, AUSTRALIA}
% \affiliation{School of Engineering \& Applied Sciences, Harvard University, Cambridge, Massachusetts 02138, USA}
\author{Youle Wang}
\affiliation{Centre for Quantum Software and Information, Faculty of Engineering and Information Technology, University of Technology Sydney, Sydeny, NSW 2007, AUSTRALIA}

\author{Yu Luo}
\affiliation{Centre for Quantum Software and Information, Faculty of Engineering and Information Technology, University of Technology Sydney, Sydeny, NSW 2007, AUSTRALIA}

\affiliation{College of Computer Science, Shaanxi Normal University, Xi'an, 710062, China}

\author{Yuan Feng}
\affiliation{Centre for Quantum Software and Information, Faculty of Engineering and Information Technology, University of Technology Sydney, Sydeny, NSW 2007, AUSTRALIA}

% \date{\today}

\begin{abstract}
We propose a quantum data fitting algorithm for non-sparse matrices, which is based on the Quantum Singular Value Estimation (QSVE) subroutine and a novel efficient method for recovering the signs of eigenvalues. Our algorithm generalizes the quantum data fitting algorithm of Wiebe, Braun, and Lloyd for sparse and well-conditioned matrices by adding a regularization term to avoid the over-fitting problem, which is a very important problem in machine learning. 
As a result, the algorithm achieves a sparsity-independent runtime of
$O(\kappa^2\sqrt{N}\mathrm{polylog}(N)/(\epsilon\log\kappa))$ for an $N\times N$ dimensional Hermitian matrix $\bm{F}$, where $\kappa$ denotes the condition number of $\bm{F}$ and $\epsilon$ is the precision parameter. This amounts to a polynomial speedup on the dimension of matrices when compared with the classical data fitting algorithms, and a strictly less than quadratic dependence on $\kappa$.
\end{abstract}

\maketitle

\section{Introduction}

Quantum machine learning is an emerging research area in the intersection of quantum computing and machine learning \cite{biamonte2017quantum,wittek2014quantum}. {In recent years, a number of quantum machine learning algorithms have been proposed, most of which could provide polynomial, sometimes exponential, speedup when compared with classical machine learning algorithms. }
%springing up to exponentially or at least polynomially improve the classical machine learning algorithms by using the advantages of quantum computation. 
This trend began with the breakthrough quantum algorithm of Harrow, Hassidim and Lloyd (HHL) \cite{harrow2009quantum} which solves a linear system $\bm{Ax}=\bm{b}$ 
with exponential acceleration over classical algorithms when the matrix $\bm{A}$ is sparse and well conditioned. More importantly, (revised versions of) HHL has been employed as a subroutine by many quantum machine learning algorithms in solving problems such as Quantum Support Vector Machine (QSVM) \cite{rebentrost2014quantum}, Quantum Recommendation Systems \cite{kerenidis2016quantum}, and so on \cite{schuld2016prediction,wiebe2014quantum,wiebe2014quantum1,kapoor2016quantum,zhao2015quantum,lloyd2014quantum,low2014quantum,rebentrost2016quantum,ciliberto2018quantum}.

In this paper, we are concerned with the Quantum Data Fitting (QDF) problem, 
whose goal is to find a quantum state proportional to the optimal fit parameter of the least squares fitting problem. It was shown in \cite{wiebe2012quantum} that by applying HHL algorithm, QDF problem can be solved in time $O(s^3\kappa^6 \log(N)/\epsilon)$, where $N$, $s$ and $\kappa$ denote the dimension, sparsity (the maximum number of nonzero elements in any given row or column), and condition number of $\bm{A}$, respectively, and $\epsilon$ is the maximum allowed distance between the output quantum state and the exact solution. Although the running time could be improved to $O(s\kappa^6\log(N)/\epsilon^2)$ via the simulation method of \cite{childs2010relationship,berry2009black} or $O(s^2\kappa^3\log(N)/\epsilon^2)$ using the method of \cite{liu2015fast}, the dependence over $s$ is at least linear, leading to a running time of at least $O(N\log(N))$ for non-sparse matrices. Hence, it remains open whether it is possible, and how, to decrease the the dependence on $N$ for non-sparse matrices in solving QDF problems.

%Another issue with the QDF algorithm is the over-fitting problem \cite{hawkins2004problem}, i.e., while the fitting effect of existing data is significantly good, the prediction effect of future data remains very poor. The QDF algorithm as described in \cite{wiebe2012quantum} does not address this problem. 

Another issue not addressed by the QDF algorithm proposed in \cite{wiebe2012quantum} is the over-fitting problem \cite{hawkins2004problem}; i.e., in some cases while the fitting of existing data is significantly good, the prediction of future data may remain poor.
In this paper, we consider the generalized standard technique for data fitting \footnote{
In machine learning, least squares (LSQ) fitting is a standard technique for data fitting and often used interchangeably with the term data fitting.}, i.e., the regularized least squares fitting, also known as the ridge regression \cite{hoerl1970ridge}, by adding a regularization term to avoid the over-fitting problem.
We propose a quantum data fitting algorithm for regularized least squares fitting problems with non-sparse matrices, with a running time of $O(\kappa^2\sqrt{N}\mathrm{polylog}(N)/(\epsilon\log\kappa))$, a polynomial speedup (on the dimension $N$) over classical algorithms.
The main result is given in Theorem~\ref{thm:main_result}.

\paragraph*{Related Works.}

Recently, inspired by the quantum recommendation systems and based on the Quantum Singular Value Estimation (QSVE) subroutine \cite{kerenidis2016quantum}, Wossnig, Zhao and Prakash (WZP) \cite{wossnig2018quantum} proposed a dense version of HHL.  
Recall that QSVE can only estimate the magnitude, but not the sign, of the eigenvalues of a Hermitian matrix. 
The key technique of WZP algorithm is to first call QSVE subroutines for matrices $\bm{A}$ and $\bm{A}+\mu \bm{I}$, respectively, where $\mu=1/\kappa$ is a relatively small number, and then compare the corresponding eigenvalues of these matrices to obtain the desired sign.

However, this technique has two potential disadvantages: 1) we need to construct two, instead of one, binary tree data structures as proposed in \cite{kerenidis2016quantum}. Constructing these binary trees is time-consuming; it is linear in the number of non-zero elements of the matrix; 2) it becomes difficult to implement if $\kappa$ is significantly large, as a small $\mu$ requires a high precision quantum computer to process.
By comparison, in this paper, we recover the signs of eigenvalues of $\bm{A}$ by using only one binary tree data structure for the matrix $\bm{A}+\|\bm{A}\|_*\bm{I}$, where $\|\bm{A}\|_*$ denotes the spectral norm of $\bm{A}$. Furthermore, we do not need to perform the comparison operation, which might introduce additional errors to the system.

It is worth noting that recently, Meng et al. \cite{meng2018quantum} and Yu et al. also \cite{yu2017quantum} proposed quantum ridge regression algorithms in the non-sparse cases.
However, the algorithm in \cite{yu2017quantum} only works for low-rank matrices, while that in \cite{meng2018quantum} uses the same technique as WZP, thus having the same potential disadvantages as we pointed out above. Moreover, neither of them explore the impact of the hyper-parameter on the time complexity of the algorithm, like we do in the current paper.

% \red{Furthermore, as we know, it is a characteristic feature of Singular Value Decomposition (SVD) that any increase of the dimensions of a matrix will increase the number of small singular values \cite{hansen2005rank}.
% This implies that, in the least squares fitting problem, whenever a new sample point or an additional row is added, while the largest singular value is nearly unchanged, the smallest singular value may become smaller. Therefore, we can assume the largest singular value, i.e., the spectral norm of the matrix, is bounded, which means our method is robust when applied to a dynamic incremental dataset.}

\section{Regularized Least Squares Fitting}\label{sec:LSF}

The least squares fitting problem \cite{wiebe2012quantum} can be described as follows. Given a set of $m$ samples $\{(x_i,y_i)\in \mathbb{C}^2 : 1\leq i\leq m\}$ \footnote{Here, we, following \cite{wiebe2012quantum}, consider the case that the data points are scalar. However, if they are more general, e.g., vectors, then we can let each function $f_j(\cdot)$ be equaling to each element of the vector, to match the more common description of the least squares fitting problem.}, the goal is to find a parametric function $f(x, \bm{w}) : \mathbb{C}^{n+1} \rightarrow \mathbb{C}$ to well approximate these points, where $\bm{w}\in \mathbb C^n$ is the fit parameter. We assume that $f(x, \bm{w})$ is linear in $\bm{w}$, but not necessarily so in $x$. In other words, 
\begin{align}
    f(x,\bm{w}):=\sum_{j=1}^n w_jf_j(x)
\end{align}
for some functions $f_j: \mathbb C^{n}\rightarrow \mathbb C$. The objective is to minimize the sum of the distance between the fit function and the target outputs $\bm{y}$ and a regularization term, i.e.,
\begin{align}\label{eq:df:objective}
    \min_{\bm{w}} \sum_{i=1}^m |f(x_i,\bm{w})-y_i|^2 +\gamma \bm{w}^\dag \bm{w} = \|\bm{Fw}-\bm{y}\|^2+ \gamma \|\bm{w}\|^2,
\end{align}
where $\bm{F} =  (f_j(x_i))_{i,j}$ is an $m\times n$ matrix, $\bm{y}=(y_1,y_2,\ldots,y_m)^\top$, and $\gamma>0$ denotes the hyper-parameter of the regularization term which is a common technique in machine learning.
% In addition, we assume, following HHL \cite{harrow2009quantum}, without loss of generality that $\frac{1}{\kappa}\leq \sigma(\bm{F})\leq 1$, where $\kappa$ and $\sigma$ denote the condition number and the singular value of $\bm{F}$, respectively. 
In this paper, we assume $\bm{F}$ is given, and our task is to find the optimal $\bm{w}$. 
The solution to the regularized least squares fitting problem (\ref{eq:df:objective}) is given by 
\begin{align}\label{eq:df:solution_original}
    \bm{w}^*=(\bm{F}^\dag \bm{F}+\gamma \bm{I}_n)^{-1}\bm{F}^\dag \bm{y},
\end{align}
where $\bm{I}_n$ denotes the $n$-by-$n$ identity matrix.

Note that we can assume without loss of generality that the matrix $\bm{F}$ is Hermitian. Otherwise, define 
    $\tilde{\bm{F}} =\tilde{\bm{F}}^\dag :=\left[\begin{array}{cc}
    \bm{0} & \bm{F} \\
    \bm{F}^\dag	 & \bm{0}
	\end{array}
	\right]$ and $\tilde{\bm{y}} := \left[\begin{array}{c}
    \bm{y}  \\   \bm{0}    \end{array}\right]\in \mathbb C^{m+n}$. Then it is easy to check that $\bm{w}^*$ satisfies Eq.~(\ref{eq:df:solution_original}) if and only if
\begin{align}\label{eq:df:solution}
    \tilde{\bm{w}}^*=(\tilde{\bm{F}}^\dag \tilde{\bm{F}}+\gamma \bm{I}_{m+n})^{-1}\tilde{\bm{F}}^\dag \tilde{\bm{y}},
\end{align}
where $\tilde{\bm{w}}^*:=\left[\begin{array}{c}
    \bm{0}  \\   \bm{w}^*    \end{array}\right]\in \mathbb C^{m+n}$.
In other words, for any non-Hermitian matrix, we can construct a Hermitian matrix which gives the same optimal solution in Eq.~(\ref{eq:df:solution_original}) by expanding the vector's dimension \cite{harrow2009quantum}.

\section{Quantum Singular Value Estimation}
Quantum Singular Value Estimation (QSVE) can be viewed as extending Phase Estimation \cite{kitaev1995quantum} from unitary to nonunitary matrices, which is also the primary algorithm subroutine for our quantum data fitting algorithm. We briefly state it in the following:

Given a matrix $\bm{A}\in \mathbb R^{m\times n}$ which is stored in a classical binary tree data structure, an algorithm having quantum access to the data structure can create, in time polylog$(mn)$, the quantum state $\ket{A_i}$ corresponding to each row $A_i$ of the matrix $\bm{A}$ \cite{kerenidis2016quantum}. Note also that if each element $A_{ij}$ of $\bm{A}$ is a complex number, the binary tree just stores its squared length $|A_{ij}|^2$ in each leaf node.

\begin{thm}\label{thm:df:qsve}
Quantum Singular Value Estimation \cite{kerenidis2016quantum}: Let $\bm{A}\in \mathbb R^{m\times n}$ be a matrix stored in the data structure presented above, and $\bm{A}=\sum_i\sigma_i \bm{u}_i\bm{v}_i^\dag$ be its singular value decomposition. For a precision parameter $\epsilon > 0$, there is a quantum algorithm that performs the mapping $\ket{\psi}=\sum_{j}\beta_j\ket{v_j}\rightarrow \sum_{j}\beta_j\ket{v_j}\ket{\overline{\sigma_j}}$ such that $|\overline{\sigma_j}-\sigma_j|\leq \epsilon$ for all $j$ with probability at least $1-1/\mathrm{poly}(mn)$ in time $O(\|\bm{A}\|_F\mathrm{polylog}(mn)/\epsilon)$.
\end{thm}

We see from Theorem~\ref{thm:df:qsve} that the runtime of QSVE depends on the Frobenius norm $\|\bm{A}\|_F$, rather than the sparsity $s(\bm{A})$ shown in HHL. This will also appear in our algorithm's runtime.

\section{Quantum Data Fitting Algorithm}

For a Hermitian matrix $\bm{F}\in \mathbb C^{N\times N}$ with the spectral decomposition $\bm{F}=\sum_i\lambda_i\bm{v}_i\bm{v}_i^\dag$, its singular value decomposition is given by $\bm{F}=\sum_i|\lambda_i|\bm{u}_i\bm{v}_i^\dag$, where the left singular vectors $\bm{u}_i$ are equal to $\pm\bm{v}_i$ depending on the signs of $\lambda_i$; i.e., $\bm{u}_i=-\bm{v}_i$ if $\lambda_i<0$, and $\bm{u}_i=\bm{v}_i$ otherwise.

Similar to~\cite{wossnig2018quantum}, QSVE in Theorem~\ref{thm:df:qsve} will also serve as a key subroutine of our algorithm. The difference is, however, we are going to use the following lemma to recover the sign of eigenvalues of a Hermitian matrix.

\begin{lem}\label{lem:qdf:recover_signs}
Let $\bm{F}\in \mathbb C^{N\times N}$ be a Hermitian matrix with the spectral decomposition $\bm{F}=\sum_i\lambda_i\bm{v}_i\bm{v}_i^\dag$. Let $\|\bm{F}\|_*=\max_{i\in [N]}\{|\lambda_i|\}$ be the spectral norm of $\bm{F}$, and $\bm{I}_N$ the $N$-by-$N$ identity matrix. For a precision parameter $\epsilon>0$, by performing QSVE algorithm on the matrix $\hat{\bm{F}}:=\bm{F}+\|\bm{F}\|_*\bm{I}_N$, we can transform  $\ket{\psi}=\sum_{j}\beta_j\ket{v_j}$ into  
$\sum_{j}\beta_j\ket{v_j}\ket{\overline{\lambda_j}}$ such that $|\overline{\lambda_j}-\lambda_j|\leq \epsilon$ for all $j\in[N]$ with probability at least $1-1/\mathrm{poly}(N)$ in time $O(\sqrt{N}\|\bm{F}\|_*\mathrm{polylog}(N)/\epsilon)$.
\end{lem}

\begin{proof}
The proof is quite straightforward. Since $\bm{F}$ has the spectral decomposition $\sum_i\lambda_i\bm{v}_i\bm{v}_i^\dag$,  $\hat{\bm{F}}=\sum_i\hat{\lambda}_i\bm{v}_i\bm{v}_i^\dag$, where $\hat{\lambda}_i=\lambda_i+\|\bm{F}\|_*$ for all $i\in[N]$. By the definition of $\|\bm{F}\|_*$, eigenvalues $\hat{\lambda}_i$ of $\hat{\bm{F}}$ are all non-negative, meaning that $\hat{\bm{F}}$ is a positive semi-definite matrix. Therefore, the singular value decomposition of $\hat{\bm{F}}$ is the same as its spectral decomposition.

By performing QSVE on $\hat{\bm{F}}$ with the precision parameter $\epsilon>0$, we obtain an estimation $\overline{\hat{\lambda}_i}$ of $\hat{\lambda}_i$ such that $|\overline{\hat{\lambda}_i}-\hat{\lambda}_i|\leq \epsilon$ for all $i\in[N]$, with probability at least $1-1/\mathrm{poly}(N)$ in time $O(\|\hat{\bm{F}}\|_F\mathrm{polylog}(N)/\epsilon)$. An estimation $\overline{{\lambda}_i}$ of $\lambda_i$ of the original matrix $\bm{F}$ is then obtained by subtracting $\|\bm{F}\|_*$ from  $\overline{\hat{\lambda}_i}$. Finally, the estimation error can be bounded the same as QSVE, because we have
\begin{align}
|\overline{\lambda_i}-\lambda_i|=|(\overline{\lambda_i}+\|\bm{F}\|_*)-(\lambda_i+\|\bm{F}\|_*)|=|\overline{\hat{\lambda}_i}-\hat{\lambda}_i|\leq \epsilon.
\end{align}

Now we consider the bound of $\|\hat{\bm{F}}\|_F$ to bound the time complexity.
\begin{align}
    \|\hat{\bm{F}}\|_F & =\sqrt{\sum_i(\hat{\lambda}_i)^2} =\sqrt{\sum_i\left(\lambda_i+ \|\bm{F}\|_*\right)^2} \nonumber\\
    &=\sqrt{\sum_i\lambda_i^2 +2\sum_i\lambda_i\|\bm{F}\|_*+N\|\bm{F}\|_*^2}  \nonumber\\ \label{eq:qdf:mean_lambda}
    &= \sqrt{\|\bm{F}\|_F^2+ \left(1+2\frac{\mathbb E (\lambda)}{\|\bm{F}\|_*}\right)N\|\bm{F}\|_*^2} \\
    &\le O(\sqrt{N}\|\bm{F}\|_*),\label{eq:qdf:F_norm_less_spectum}
\end{align}
where $\mathbb E (\lambda):=\frac{1}{N}\sum_i\lambda_i\in\left[-\|\bm{F}\|_*,\|\bm{F}\|_*\right]$ and Eq.(\ref{eq:qdf:F_norm_less_spectum})
follows from $\|\bm{F}\|_F\le \sqrt{N}\|\bm{F}\|_*$.
This completes the proof.
\end{proof}

With this lemma, we propose our quantum data fitting algorithm as in the following theorem:

\begin{thm}\label{thm:main_result}
Let $\bm{F}\in \mathbb C^{N\times N}$ be the non-sparse Hermitian matrix described in the least squares fitting problem, $\bm{F}=\sum_i\lambda_i\bm{v}_i\bm{v}_i^\dag$ its spectral decomposition, 
and $\kappa$ its condition number.
Assume that $\hat{\bm{F}}=\bm{F}+\|\bm{F}\|_*\bm{I}_N$ is stored in the classical binary tree data structure as in \cite{kerenidis2016quantum}.  For a precision parameter $\epsilon>0$, Algorithm~\ref{alg:qdf:qdfa_for_non_sm} outputs a quantum state \ket{w} such that $\|\ket{w}-\ket{{w}^*}\|\leq\epsilon$ in time $O\left(\kappa^2\sqrt{N}\mathrm{polylog}(N)/(\epsilon\log\kappa)\right)$, where $\ket{{w}^*}$ denotes the quantum state proportional to the optimal fit parameter $\bm{w}^*$ in Eq.~(\ref{eq:df:solution_original}).
% \yf{with success probability $O(1)$?}
\end{thm}

\begin{algorithm}[t]
\raggedright
\caption{Quantum Data Fitting Algorithm for Non-Sparse Matrices} \label{alg:qdf:qdfa_for_non_sm}
\KwIn{$\hat{\bm{F}}=\bm{F}+\|\bm{F}\|_*\bm{I}_N\in\mathbb C^{N\times N}$ and $\bm{y}\in\mathbb C^N$ stored in the classical binary tree data structure required by QSVE~\cite{kerenidis2016quantum}, the condition number $\kappa$ (or an upper bound of it) of $\bm{F}$, and precision $\epsilon$.}
\KwOut{A quantum state $\ket{w}$ which is proportional to the optimal fit parameter $\bm{w}^*$ with the bounding error $\epsilon$ as measured by the Euclidean distance.}
\begin{enumerate}
\item Generate a value of hyper-parameter $\gamma\in [\frac{\|\bm{F}\|_*^2}{\kappa^2},\|\bm{F}\|_*^2] $ according to the log-uniform distribution. \label{alg:qdf:line_1}
\item Create the quantum state $\ket{y}=\sum_{i}\beta_i\ket{v_i}$ which is proportional to $\bm{y}$, with $\bm{v}_i$'s being the eigenvectors of $\bm{F}$.
\item Perform the QSVE subroutine for matrix $\hat{\bm{F}}$ with precision $\delta=\frac{\|\bm{F}\|_*}{4\kappa}\epsilon$ to obtain the state
$\sum_{i}\beta_i\ket{v_i}\ket{\overline{\hat{\lambda}_i}}$.
\item Add an auxiliary register and apply a rotation conditioned on the second register, and uncompute the QSVE subroutine to erase the second register, obtaining
\begin{align}\label{eq:qdf:rotation}
\sum_{i}\beta_i\ket{v_i}\left(\frac{C\overline{\lambda_i}}{\overline{\lambda_i}^2+\gamma}\ket{0}+\sqrt{1-\left(\frac{C\overline{\lambda_i}}{\overline{\lambda_i}^2+\gamma}\right)^2}\ket{1}\right),
\end{align}
where $\overline{\lambda_i}=\overline{\hat{\lambda}_i}-\|\bm{F}\|_*$ is the estimation of the eigenvalue $\lambda_i$ of $\bm{F}$ and $C=C_0\sqrt{\gamma}$ ($0< C_0<2$) is a constant.
\item Post-select on the auxiliary register being in state $\ket{0}$.
% \yf{Repeat the algorithm up to $O(\kappa/\log\kappa)$ times until success?}
\end{enumerate}
\end{algorithm}

\begin{proof}
The proof mainly contains correctness analysis and complexity analysis. First we give the proof of correctness, i.e., $\|\ket{w}-\ket{w^*}\|\leq\epsilon$. 

From Algorithm~\ref{alg:qdf:qdfa_for_non_sm}, we observe, after post-selection, that
\begin{align}
    \ket{w}=\frac{\sum_{i}\beta_ih(\overline{\lambda_i})\ket{v_i}\ket{0}}{\sqrt{\sum_{i}|\beta_i|^2\left(h(\overline{\lambda_i})\right)^2}}=\frac{\sum_{i}\beta_ih(\overline{\lambda_i})\ket{v_i}\ket{0}}{\sqrt{\overline{p}}},
\end{align}
where $\overline{p}:= \sum_{i}|\beta_i|^2\left(h(\overline{\lambda_i})\right)^2$ and $h$ is defined as follows:
\begin{equation}\label{eq:qdf:h_define}
    h(\lambda):=\frac{C\lambda}{\lambda^2+\gamma}=\frac{\sqrt{\gamma}\lambda}{\lambda^2+\gamma}.
\end{equation}
Here, we take $C=C_0\sqrt{\gamma}=\sqrt{\gamma}$ as an example (Other values of $C_0\in(0,2)$ are similar).
The ideal state $\ket{w^*}$ should be 
\begin{align}
    \ket{w^*}=\frac{\sum_{i}\beta_ih({\lambda_i})\ket{v_i}\ket{0}}{\sqrt{\sum_{i}|\beta_i|^2\left(h({\lambda_i})\right)^2}}=\frac{\sum_{i}\beta_ih({\lambda_i})\ket{v_i}\ket{0}}{\sqrt{p}},
\end{align}
where $p:=  \sum_{i}|\beta_i|^2\left(h({\lambda_i})\right)^2$. Therefore, we have
\begin{align}
    &\|\ket{w}-\ket{w^*}\|^2  \nonumber\\
    % =&\|\frac{\sum_{i}\beta_ih(\overline{\lambda_i})\ket{v_i}\ket{0}}{\sqrt{\overline{p}}} - \frac{\sum_{i}\beta_ih({\lambda_i})\ket{v_i}\ket{0}}{\sqrt{p}}\|^2 \\
    &=  \|\sum_{i}\beta_i\left(\frac{h(\overline{\lambda_i})}{\sqrt{\overline{p}}} - \frac{h({\lambda_i})}{\sqrt{p}}\right)\ket{v_i}\ket{0}\|^2 \\
    &=  \sum_{i}|\beta_i|^2 \left(\frac{h({\lambda_i})}{\sqrt{p}}\right)^2 \left(\frac{h(\overline{\lambda_i})}{h({\lambda_i})}\cdot \frac{\sqrt{p}}{\sqrt{\overline{p}}} - 1 \right)^2.   \label{eq:qdf:w_w*}
\end{align}

\begin{figure*}[t]
\centering
\subfigure[small $\gamma$]{\label{fig:h_lambda1}
 \includegraphics[width=0.32\textwidth]{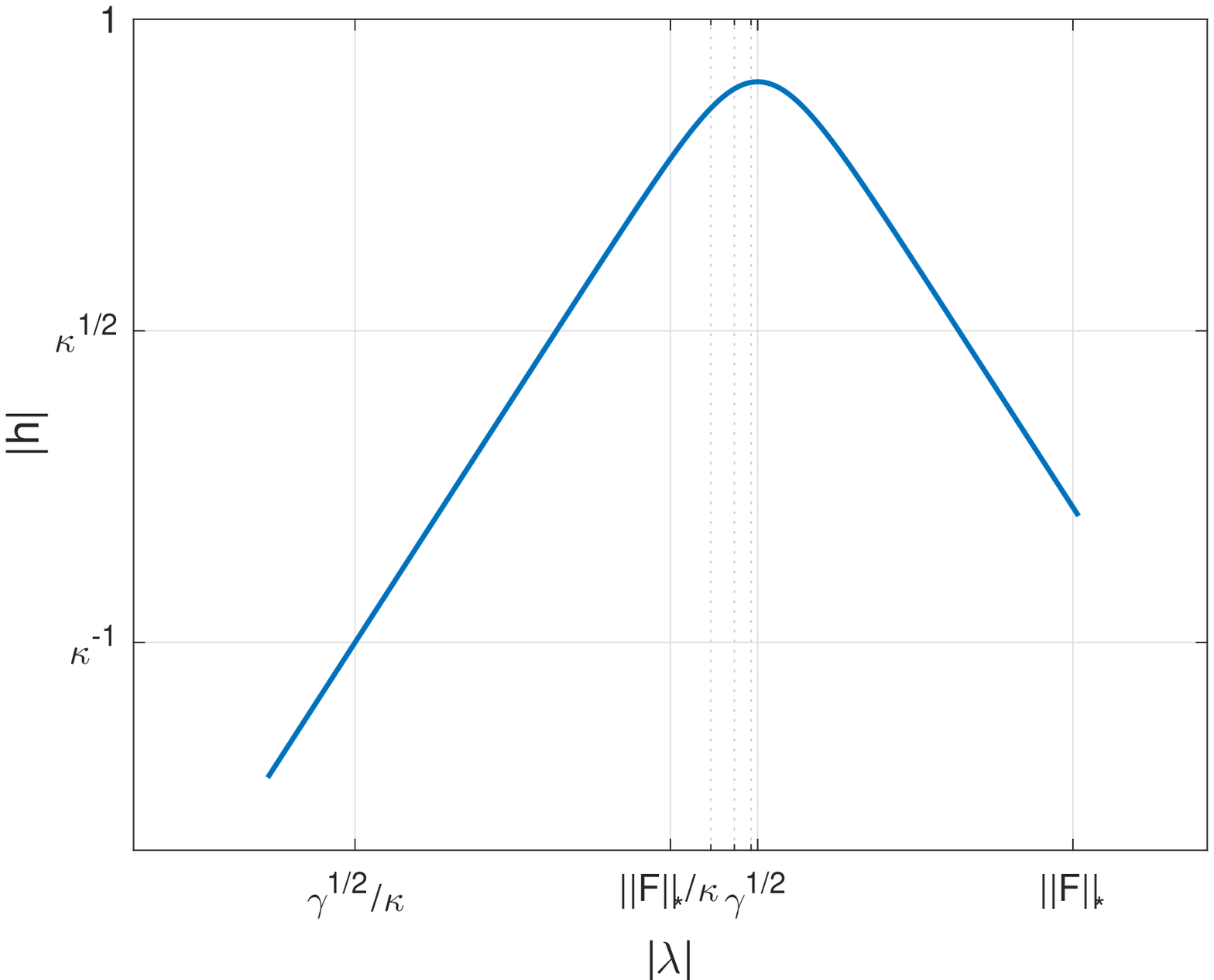}}
 \subfigure[median $\gamma$]{\label{fig:h_lambda2}
 \includegraphics[width=0.32\textwidth]{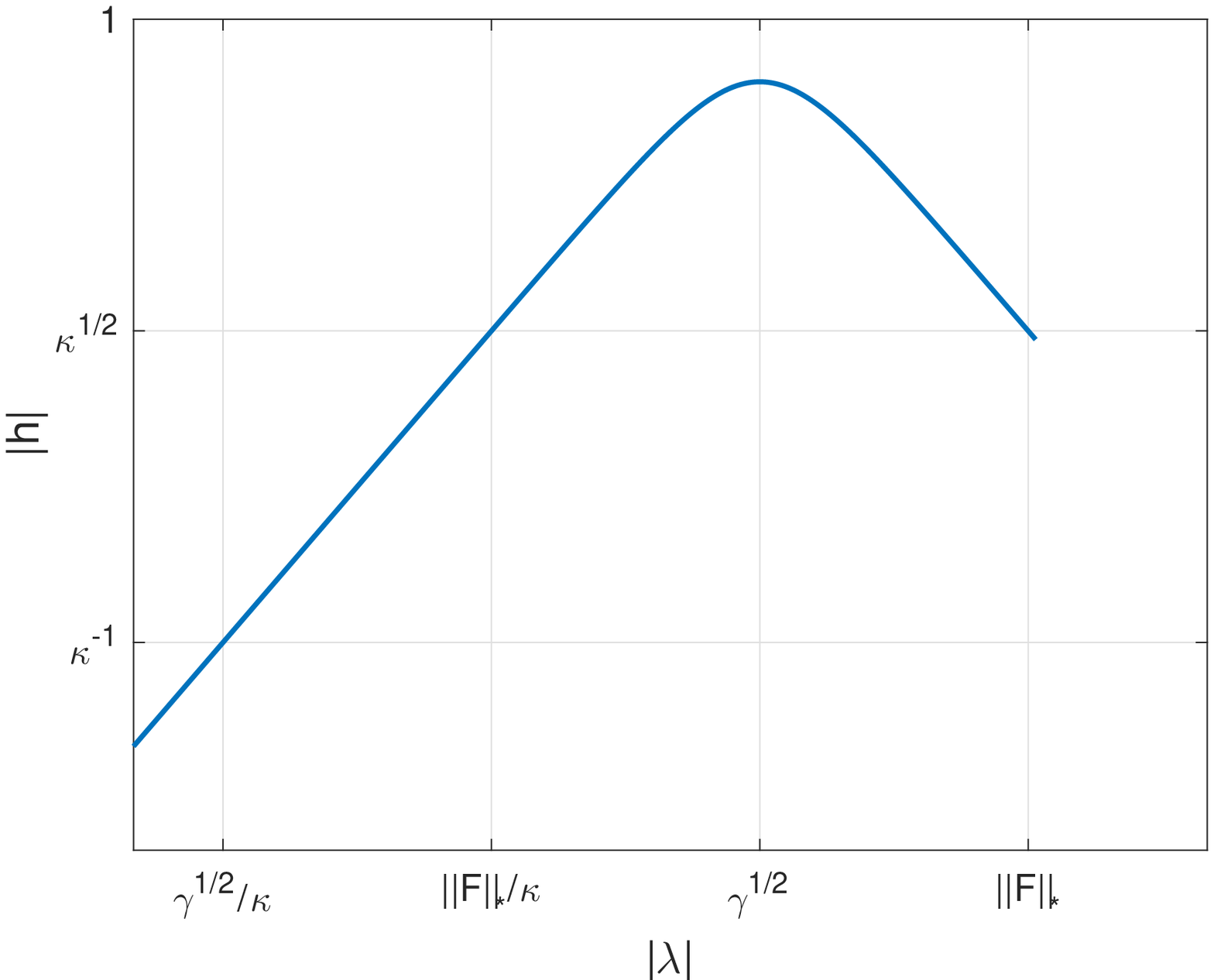}}
 \subfigure[large $\gamma$]{\label{fig:h_lambda3}
 \includegraphics[width=0.32\textwidth]{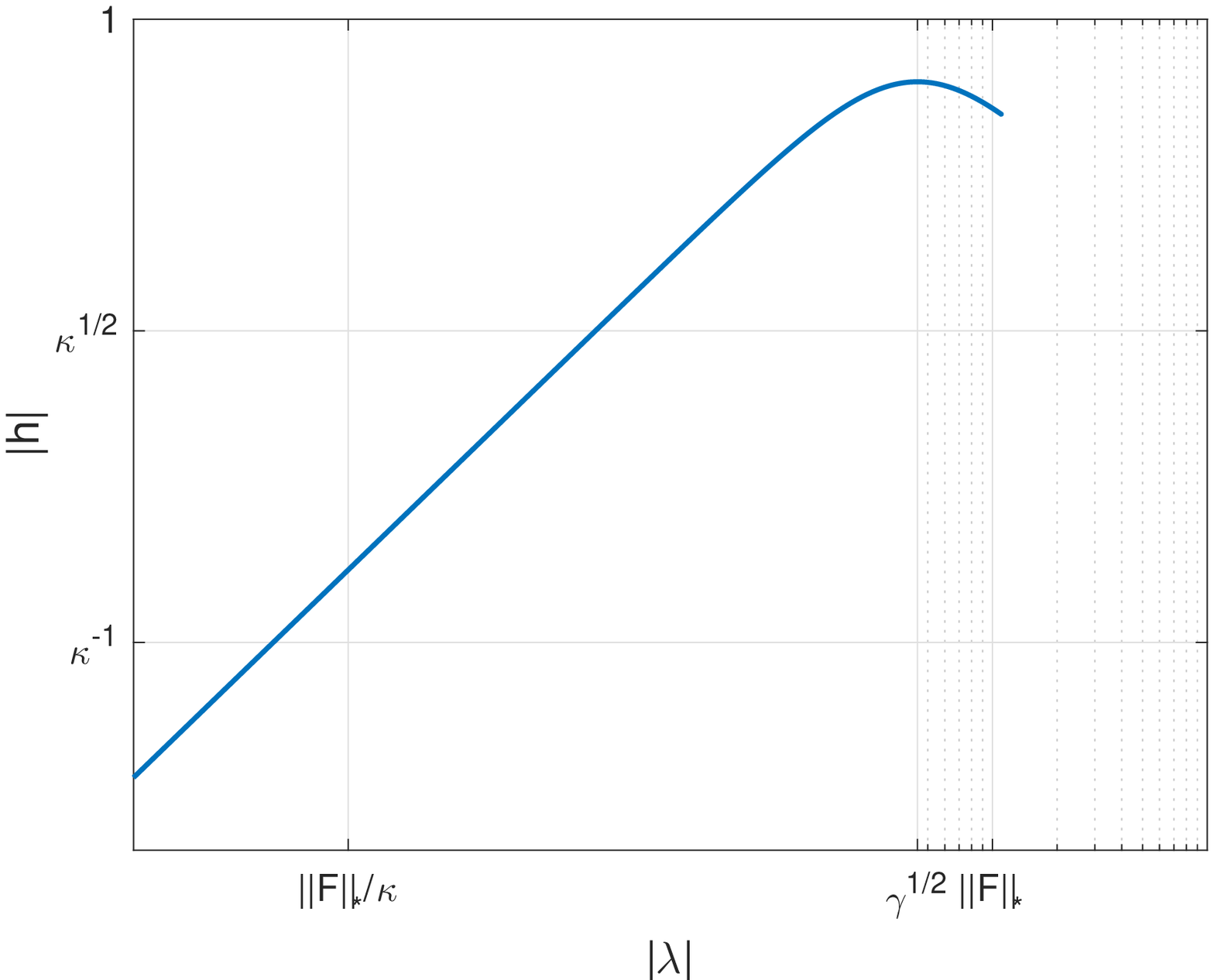}}
%  \vspace{-2ex}
  \caption{$|h|$ versus $|\lambda|$ with different $\gamma$.}
  \label{fig:h_lambda}
 \end{figure*}

We now bound $\frac{h(\overline{\lambda_i})}{h({\lambda_i})}$ and $\frac{\sqrt{p}}{\sqrt{\overline{p}}}$ via the following lemma:
\begin{lem}\label{lem:qdf:h_bar_minus_h}
Let $h(\lambda)$ be defined as in (\ref{eq:qdf:h_define}).
Then
\begin{align}
 \left|h(\overline{\lambda})-h({\lambda})\right|\leq  \frac{1}{3}\epsilon\left|h(\lambda)\right|.
\end{align}
% \yf{But what are $\lambda$ and $\overline{\lambda}$ here?}
\end{lem}
\begin{proof}
From the definition of $h$ and the fact that $|\overline{\lambda_i}-\lambda_i|\leq\delta = \frac{\|\bm{F}\|_*}{4\kappa}\epsilon$ and $|\lambda_i|\ge \frac{\|\bm{F}\|_*}{\kappa}$ for all $i$, we have
\begin{align}
    &\left|h(\overline{\lambda})-h({\lambda})\right|  \nonumber\\
    &=\sqrt{\gamma}\left|\frac{ \overline{\lambda}}{\overline{\lambda}^2 +\gamma}-\frac{\lambda}{\lambda^2+ \gamma}\right| = \sqrt{\gamma}\frac{\left|(\overline{\lambda}{\lambda} -\gamma)(\overline{\lambda}-{\lambda})\right|}{(\overline{\lambda}^2 +\gamma)({\lambda}^2+ \gamma)}  \nonumber\\\label{eq:qdf:lambda_bar_range}
    &\le \sqrt{\gamma}\frac{\left|({(1+\frac{\epsilon}{4})\lambda^2} -\gamma)(\overline{\lambda}-{\lambda})\right|}{((1-\frac{\epsilon}{4}){\lambda}^2 +\gamma)({\lambda}^2+ \gamma)}  \\ \label{eq:qdf:estimate_constant}
    &\le \frac{4}{3}\cdot\frac{1}{|\lambda|}\cdot\left|\overline{\lambda}-{\lambda}\right|\cdot\frac{\sqrt{\gamma}|\lambda|}{\lambda^2+ \gamma} \\
    &\le \frac{4}{3}\cdot\frac{\kappa}{\|\bm{F}\|_*}\cdot\delta\cdot\left|h(\lambda)\right| \\
    &=\frac{1}{3}\epsilon\left|h(\lambda)\right|,
\end{align}
where Eq.~(\ref{eq:qdf:lambda_bar_range}) follows from $(1-\frac{\epsilon}{4})|\lambda|\le|\lambda|-\delta\le|\overline{\lambda}|\le|\lambda|+\delta\le(1+\frac{\epsilon}{4})|\lambda|$ and Eq.~(\ref{eq:qdf:estimate_constant}) from $\frac{1+\epsilon/4}{1-\epsilon/4}\le 4/3$. This completes the proof of Lemma~\ref{lem:qdf:h_bar_minus_h}.
\end{proof}

From Lemma~\ref{lem:qdf:h_bar_minus_h}, we can obtain that for all $i$
\begin{align}
     \left|\frac{h(\overline{\lambda_i})}{h({\lambda_i})} -1\right| \le \frac{1}{3}\epsilon.
\end{align}
Thus
\begin{align} \label{eq:qdf:h_h}
   1-\frac{1}{3}\epsilon\le  \frac{h(\overline{\lambda_i})}{h({\lambda_i})} \le 1+\frac{1}{3}\epsilon.
\end{align}
And then 
\begin{align} \label{eq:qdf:p_p}
    \frac{1}{1+\frac{1}{3}\epsilon} \le \frac{\sqrt{p}}{\sqrt{\overline{p}}}= \sqrt{\frac{\sum_{i}|\beta_i|^2\left(h({\lambda_i})\right)^2}{\sum_{i}|\beta_i|^2\left(h(\overline{\lambda_i})\right)^2}} \le \frac{1}{1-\frac{1}{3}\epsilon}.
\end{align}
By substituting (\ref{eq:qdf:h_h}) and (\ref{eq:qdf:p_p}) into (\ref{eq:qdf:w_w*}), we have 
\begin{align}
    \|\ket{w}-\ket{w^*}\| &\le \max\left\{1-\frac{1-\frac{1}{3}\epsilon}{1+\frac{1}{3}\epsilon}, \frac{1+\frac{1}{3}\epsilon}{1-\frac{1}{3}\epsilon}-1\right\} \nonumber \\
    &=\frac{\frac{2}{3}\epsilon}{1-\frac{1}{3}\epsilon} \le \epsilon.
\end{align}

Next we give the proof of the time complexity. 
From Lemma~\ref{lem:qdf:recover_signs}, we know that in Algorithm~\ref{alg:qdf:qdfa_for_non_sm}, the QSVE subroutine runs in time
\begin{align}\label{eq:qsve_runtime}
    O(\sqrt{N}\|\bm{F}\|_*\mathrm{polylog}(N)/\delta)=O(\kappa\sqrt{N}\mathrm{polylog}(N)/\epsilon).
\end{align}
On the other hand, we consider the success probability of the post-selection process. In order to bound the maximal number of iterations, we need to compute the minimum of the rotation function $h(\lambda)$ which is related to the hyper-parameter $\gamma$. 
The image of $|h|$ as a function of $\gamma$ is illustrated in Figure~\ref{fig:h_lambda} \footnote{In general, we take $\gamma\in [\frac{\|\bm{F}\|_*^2}{\kappa^2},\|\bm{F}\|_*^2]$. This is reasonable in machine learning area because too small values of $\gamma$ lead to a negligible effect of regularization while too large values of $\gamma$ result in the loss of useful information of the original problems, i.e., the so-called under-fitting \cite{svergun1992determination}.}, from which we see  that, for $|\lambda|\in\left[\frac{\|\bm{F}\|_*}{\kappa},\|\bm{F}\|_*\right]$, $\min |h|$ is given by
\begin{align}\label{eq:qdf:h_gamma}
\renewcommand*{\arraystretch}{2}
   \left\{\begin{array}{ll}
    h(\frac{\|\bm{F}\|_*}{\kappa}) =\frac{\sqrt{\gamma} \frac{\|\bm{F}\|_*}{\kappa}}{\frac{\|\bm{F}\|_*^2}{\kappa^2}+\gamma}  \geq \frac{\|\bm{F}\|_*}{2\sqrt{\gamma}\kappa}   \text{,} & \text{if }   \frac{\|\bm{F}\|_*^2}{\kappa}\le \gamma \le \|\bm{F}\|_*^2\text{;}   \\
    h(\|\bm{F}\|_*)  = \frac{\sqrt{\gamma}\|\bm{F}\|_*}{\|\bm{F}\|_*^2 +\gamma} \geq \frac{\sqrt{\gamma}}{2\|\bm{F}\|_*}\text{,} &
    \text{if } \frac{\|\bm{F}\|_*^2}{\kappa^2} \le \gamma < \frac{\|\bm{F}\|_*^2}{\kappa} \text{.}
    \end{array}
    \right.
\end{align}
Hence, using amplitude amplification \cite{brassard2002quantum}, the number of iterations could be bounded as $O(\max\{\frac{\sqrt{\gamma}\kappa}{\|\bm{F}\|_*},\frac{\|\bm{F}\|_*}{\sqrt{\gamma}} \})$.

Furthermore, from the experience of machine learning,  $\gamma$ is usually taken  {in a logarithmic scale}, e.g., 0.01, 0.1, 1, $\ldots$ \cite{montavon1998tricks}. Thus we take $\gamma$ randomly according to a log-uniform distribution in its domain (Line~\ref{alg:qdf:line_1} of Algorithm~\ref{alg:qdf:qdfa_for_non_sm}). We estimate the number of iterations as
\begin{align}
  \frac{1}{\ln\|\bm{F}\|_*^2- \ln\frac{\|\bm{F}\|_*^2}{\kappa^2}} \int_{\ln\frac{\|\bm{F}\|_*^2}{\kappa^2}}^{\ln\|\bm{F}\|_*^2} &\max\{\frac{\kappa\sqrt{\mathrm{e}^t}}{\|\bm{F}\|_*},\frac{\|\bm{F}\|_*}{\sqrt{\mathrm{e}^t}} \} dt  \nonumber\\
  &= O(\kappa/\log\kappa),
\end{align}
where $t=\ln\gamma$ obeys a uniform distribution.
% \yf{The logic should be like this: to achieve $O(1)$ success probability, we repeat steps 1 to 5 in Algorithm 1 up to $O(\kappa/\log\kappa)$ times (terminates earlier if the post-selection is successful in step 5). That is, we need to specify the number of iterations in the algorithm. Then here we should compute the success probability}
Combining with (\ref{eq:qsve_runtime}), the totoal time complexity is $O\left(\kappa^2\sqrt{N}\mathrm{polylog}(N)/(\epsilon\log\kappa)\right)$. This completes the proof of Theorem~\ref{thm:main_result}.
\end{proof}

\section{Further Discussions and Conclusions}

In this paper, we proposed a quantum data fitting algorithm for regularized least squares fitting problem with non-sparse matrices, which achieves a runtime of $O(\kappa^2\sqrt{N}\mathrm{polylog}(N)/(\epsilon\log\kappa))$, where the term $\log\kappa$ is due to the random choice of hyper-parameter $\gamma$ according to the log-uniform distribution in Algorithm~\ref{alg:qdf:qdfa_for_non_sm}.
% \yf{These two sentences do not make sense: From {Eq.~(\ref{eq:qdf:h_gamma})}, we find that when $\gamma$ is close to $\|\bm{F}\|_*^2/\kappa$, we get the best runtime $O(\kappa\sqrt{\kappa N}\mathrm{polylog}(N)/\epsilon)$; {and,} when $\gamma$ is close to $\|\bm{F}\|_*^2$ or $\|\bm{F}\|_*^2/\kappa^2$, we obtain the worst runtime $O(\kappa^2\sqrt{N}\mathrm{polylog}(N)/\epsilon)$.}
As the hyper-parameter $\gamma$ is usually set empirically in machine learning, we let our algorithm generate it automatically. Of course, if one wants to set it manually, he can simply modify our algorithm by moving the first line into the Input.

The technique proposed in this paper could also be applied to HHL algorithm, which would have the same time complexity as WZP \cite{wossnig2018quantum}.
It is worth noting that our algorithm's running time is actually related to the mean of the eigenvalues $\mathbb E[\lambda]$ of $\bm{F}$, see Eq.~(\ref{eq:qdf:mean_lambda}). If $\mathbb E[\lambda]$ is close to $-\|\bm{F}\|_*$ or all the eigenvalues are negative, then the running time is actually relatively small, e.g, maybe logarithmic on the matrix dimension $N$. If $\mathbb E[\lambda]=0$, as shown in the case of $\tilde{\bm{F}}$ in Eq.~(\ref{eq:df:solution}), or $\mathbb E[\lambda]>0$, then the running time is root quadratic on the matrix dimension, as stated in this paper.
However, on the whole, the time complexity of our algorithm is still polynomial in the dimension of the data matrix, because it is derived from the Frobenius norm, or more precisely, from the binary tree data structure \cite{kerenidis2016quantum}. Whether there exists a QDF algorithm which runs in logarithmic time on the dimension of non-sparse matrices is still need to be explored.

\begin{acknowledgments}
We thank Prof. Sanjiang Li for helpful discussions and proofreading the manuscript. G. Li acknowledges the financial support from China Scholarship Council (No. 201806070139). This work was partly supported by the Australian Research Council (Grant No: DP180100691) and the Baidu-UTS collaborative project ``AI meets Quantum: Quantum algorithms for knowledge representation and learning''.
\end{acknowledgments}

\bibliography{references}

%merlin.mbs apsrev4-1.bst 2010-07-25 4.21a (PWD, AO, DPC) hacked
%Control: key (0)
%Control: author (72) initials jnrlst
%Control: editor formatted (1) identically to author
%Control: production of article title (-1) disabled
%Control: page (0) single
%Control: year (1) truncated
%Control: production of eprint (0) enabled
\begin{thebibliography}{27}%
\makeatletter
\providecommand \@ifxundefined [1]{%
 \@ifx{#1\undefined}
}%
\providecommand \@ifnum [1]{%
 \ifnum #1\expandafter \@firstoftwo
 \else \expandafter \@secondoftwo
 \fi
}%
\providecommand \@ifx [1]{%
 \ifx #1\expandafter \@firstoftwo
 \else \expandafter \@secondoftwo
 \fi
}%
\providecommand \natexlab [1]{#1}%
\providecommand \enquote  [1]{``#1''}%
\providecommand \bibnamefont  [1]{#1}%
\providecommand \bibfnamefont [1]{#1}%
\providecommand \citenamefont [1]{#1}%
\providecommand \href@noop [0]{\@secondoftwo}%
\providecommand \href [0]{\begingroup \@sanitize@url \@href}%
\providecommand \@href[1]{\@@startlink{#1}\@@href}%
\providecommand \@@href[1]{\endgroup#1\@@endlink}%
\providecommand \@sanitize@url [0]{\catcode `\\12\catcode `\$12\catcode
  `\&12\catcode `\#12\catcode `\^12\catcode `\_12\catcode `\%12\relax}%
\providecommand \@@startlink[1]{}%
\providecommand \@@endlink[0]{}%
\providecommand \url  [0]{\begingroup\@sanitize@url \@url }%
\providecommand \@url [1]{\endgroup\@href {#1}{\urlprefix }}%
\providecommand \urlprefix  [0]{URL }%
\providecommand \Eprint [0]{\href }%
\providecommand \doibase [0]{http://dx.doi.org/}%
\providecommand \selectlanguage [0]{\@gobble}%
\providecommand \bibinfo  [0]{\@secondoftwo}%
\providecommand \bibfield  [0]{\@secondoftwo}%
\providecommand \translation [1]{[#1]}%
\providecommand \BibitemOpen [0]{}%
\providecommand \bibitemStop [0]{}%
\providecommand \bibitemNoStop [0]{.\EOS\space}%
\providecommand \EOS [0]{\spacefactor3000\relax}%
\providecommand \BibitemShut  [1]{\csname bibitem#1\endcsname}%
\let\auto@bib@innerbib\@empty
%</preamble>
\bibitem [{\citenamefont {Biamonte}\ \emph {et~al.}(2017)\citenamefont
  {Biamonte}, \citenamefont {Wittek}, \citenamefont {Pancotti}, \citenamefont
  {Rebentrost}, \citenamefont {Wiebe},\ and\ \citenamefont
  {Lloyd}}]{biamonte2017quantum}%
  \BibitemOpen
  \bibfield  {author} {\bibinfo {author} {\bibfnamefont {J.}~\bibnamefont
  {Biamonte}}, \bibinfo {author} {\bibfnamefont {P.}~\bibnamefont {Wittek}},
  \bibinfo {author} {\bibfnamefont {N.}~\bibnamefont {Pancotti}}, \bibinfo
  {author} {\bibfnamefont {P.}~\bibnamefont {Rebentrost}}, \bibinfo {author}
  {\bibfnamefont {N.}~\bibnamefont {Wiebe}}, \ and\ \bibinfo {author}
  {\bibfnamefont {S.}~\bibnamefont {Lloyd}},\ }\href@noop {} {\bibfield
  {journal} {\bibinfo  {journal} {Nature}\ }\textbf {\bibinfo {volume} {549}},\
  \bibinfo {pages} {195} (\bibinfo {year} {2017})}\BibitemShut {NoStop}%
\bibitem [{\citenamefont {Wittek}(2014)}]{wittek2014quantum}%
  \BibitemOpen
  \bibfield  {author} {\bibinfo {author} {\bibfnamefont {P.}~\bibnamefont
  {Wittek}},\ }\href@noop {} {\emph {\bibinfo {title} {Quantum machine
  learning: what quantum computing means to data mining}}}\ (\bibinfo
  {publisher} {Academic Press},\ \bibinfo {year} {2014})\BibitemShut {NoStop}%
\bibitem [{\citenamefont {Harrow}\ \emph {et~al.}(2009)\citenamefont {Harrow},
  \citenamefont {Hassidim},\ and\ \citenamefont {Lloyd}}]{harrow2009quantum}%
  \BibitemOpen
  \bibfield  {author} {\bibinfo {author} {\bibfnamefont {A.~W.}\ \bibnamefont
  {Harrow}}, \bibinfo {author} {\bibfnamefont {A.}~\bibnamefont {Hassidim}}, \
  and\ \bibinfo {author} {\bibfnamefont {S.}~\bibnamefont {Lloyd}},\
  }\href@noop {} {\bibfield  {journal} {\bibinfo  {journal} {Physical review
  letters}\ }\textbf {\bibinfo {volume} {103}},\ \bibinfo {pages} {150502}
  (\bibinfo {year} {2009})}\BibitemShut {NoStop}%
\bibitem [{\citenamefont {Rebentrost}\ \emph {et~al.}(2014)\citenamefont
  {Rebentrost}, \citenamefont {Mohseni},\ and\ \citenamefont
  {Lloyd}}]{rebentrost2014quantum}%
  \BibitemOpen
  \bibfield  {author} {\bibinfo {author} {\bibfnamefont {P.}~\bibnamefont
  {Rebentrost}}, \bibinfo {author} {\bibfnamefont {M.}~\bibnamefont {Mohseni}},
  \ and\ \bibinfo {author} {\bibfnamefont {S.}~\bibnamefont {Lloyd}},\
  }\href@noop {} {\bibfield  {journal} {\bibinfo  {journal} {Physical review
  letters}\ }\textbf {\bibinfo {volume} {113}},\ \bibinfo {pages} {130503}
  (\bibinfo {year} {2014})}\BibitemShut {NoStop}%
\bibitem [{\citenamefont {Kerenidis}\ and\ \citenamefont
  {Prakash}(2016)}]{kerenidis2016quantum}%
  \BibitemOpen
  \bibfield  {author} {\bibinfo {author} {\bibfnamefont {I.}~\bibnamefont
  {Kerenidis}}\ and\ \bibinfo {author} {\bibfnamefont {A.}~\bibnamefont
  {Prakash}},\ }\href@noop {} {\bibfield  {journal} {\bibinfo  {journal} {arXiv
  preprint arXiv:1603.08675}\ } (\bibinfo {year} {2016})}\BibitemShut {NoStop}%
\bibitem [{\citenamefont {Schuld}\ \emph {et~al.}(2016)\citenamefont {Schuld},
  \citenamefont {Sinayskiy},\ and\ \citenamefont
  {Petruccione}}]{schuld2016prediction}%
  \BibitemOpen
  \bibfield  {author} {\bibinfo {author} {\bibfnamefont {M.}~\bibnamefont
  {Schuld}}, \bibinfo {author} {\bibfnamefont {I.}~\bibnamefont {Sinayskiy}}, \
  and\ \bibinfo {author} {\bibfnamefont {F.}~\bibnamefont {Petruccione}},\
  }\href@noop {} {\bibfield  {journal} {\bibinfo  {journal} {Physical Review
  A}\ }\textbf {\bibinfo {volume} {94}},\ \bibinfo {pages} {022342} (\bibinfo
  {year} {2016})}\BibitemShut {NoStop}%
\bibitem [{\citenamefont {Wiebe}\ \emph
  {et~al.}(2014{\natexlab{a}})\citenamefont {Wiebe}, \citenamefont {Kapoor},\
  and\ \citenamefont {Svore}}]{wiebe2014quantum}%
  \BibitemOpen
  \bibfield  {author} {\bibinfo {author} {\bibfnamefont {N.}~\bibnamefont
  {Wiebe}}, \bibinfo {author} {\bibfnamefont {A.}~\bibnamefont {Kapoor}}, \
  and\ \bibinfo {author} {\bibfnamefont {K.}~\bibnamefont {Svore}},\
  }\href@noop {} {\bibfield  {journal} {\bibinfo  {journal} {arXiv preprint
  arXiv:1401.2142}\ } (\bibinfo {year} {2014}{\natexlab{a}})}\BibitemShut
  {NoStop}%
\bibitem [{\citenamefont {Wiebe}\ \emph
  {et~al.}(2014{\natexlab{b}})\citenamefont {Wiebe}, \citenamefont {Kapoor},\
  and\ \citenamefont {Svore}}]{wiebe2014quantum1}%
  \BibitemOpen
  \bibfield  {author} {\bibinfo {author} {\bibfnamefont {N.}~\bibnamefont
  {Wiebe}}, \bibinfo {author} {\bibfnamefont {A.}~\bibnamefont {Kapoor}}, \
  and\ \bibinfo {author} {\bibfnamefont {K.~M.}\ \bibnamefont {Svore}},\
  }\href@noop {} {\bibfield  {journal} {\bibinfo  {journal} {arXiv preprint
  arXiv:1412.3489}\ } (\bibinfo {year} {2014}{\natexlab{b}})}\BibitemShut
  {NoStop}%
\bibitem [{\citenamefont {Kapoor}\ \emph {et~al.}(2016)\citenamefont {Kapoor},
  \citenamefont {Wiebe},\ and\ \citenamefont {Svore}}]{kapoor2016quantum}%
  \BibitemOpen
  \bibfield  {author} {\bibinfo {author} {\bibfnamefont {A.}~\bibnamefont
  {Kapoor}}, \bibinfo {author} {\bibfnamefont {N.}~\bibnamefont {Wiebe}}, \
  and\ \bibinfo {author} {\bibfnamefont {K.}~\bibnamefont {Svore}},\ }in\
  \href@noop {} {\emph {\bibinfo {booktitle} {Advances in Neural Information
  Processing Systems}}}\ (\bibinfo {year} {2016})\ pp.\ \bibinfo {pages}
  {3999--4007}\BibitemShut {NoStop}%
\bibitem [{\citenamefont {Zhao}\ \emph {et~al.}(2015)\citenamefont {Zhao},
  \citenamefont {Fitzsimons},\ and\ \citenamefont
  {Fitzsimons}}]{zhao2015quantum}%
  \BibitemOpen
  \bibfield  {author} {\bibinfo {author} {\bibfnamefont {Z.}~\bibnamefont
  {Zhao}}, \bibinfo {author} {\bibfnamefont {J.~K.}\ \bibnamefont
  {Fitzsimons}}, \ and\ \bibinfo {author} {\bibfnamefont {J.~F.}\ \bibnamefont
  {Fitzsimons}},\ }\href@noop {} {\bibfield  {journal} {\bibinfo  {journal}
  {arXiv preprint arXiv:1512.03929}\ } (\bibinfo {year} {2015})}\BibitemShut
  {NoStop}%
\bibitem [{\citenamefont {Lloyd}\ \emph {et~al.}(2014)\citenamefont {Lloyd},
  \citenamefont {Mohseni},\ and\ \citenamefont
  {Rebentrost}}]{lloyd2014quantum}%
  \BibitemOpen
  \bibfield  {author} {\bibinfo {author} {\bibfnamefont {S.}~\bibnamefont
  {Lloyd}}, \bibinfo {author} {\bibfnamefont {M.}~\bibnamefont {Mohseni}}, \
  and\ \bibinfo {author} {\bibfnamefont {P.}~\bibnamefont {Rebentrost}},\
  }\href@noop {} {\bibfield  {journal} {\bibinfo  {journal} {Nature Physics}\
  }\textbf {\bibinfo {volume} {10}},\ \bibinfo {pages} {631} (\bibinfo {year}
  {2014})}\BibitemShut {NoStop}%
\bibitem [{\citenamefont {Low}\ \emph {et~al.}(2014)\citenamefont {Low},
  \citenamefont {Yoder},\ and\ \citenamefont {Chuang}}]{low2014quantum}%
  \BibitemOpen
  \bibfield  {author} {\bibinfo {author} {\bibfnamefont {G.~H.}\ \bibnamefont
  {Low}}, \bibinfo {author} {\bibfnamefont {T.~J.}\ \bibnamefont {Yoder}}, \
  and\ \bibinfo {author} {\bibfnamefont {I.~L.}\ \bibnamefont {Chuang}},\
  }\href@noop {} {\bibfield  {journal} {\bibinfo  {journal} {Physical Review
  A}\ }\textbf {\bibinfo {volume} {89}},\ \bibinfo {pages} {062315} (\bibinfo
  {year} {2014})}\BibitemShut {NoStop}%
\bibitem [{\citenamefont {Rebentrost}\ \emph {et~al.}(2016)\citenamefont
  {Rebentrost}, \citenamefont {Schuld}, \citenamefont {Wossnig}, \citenamefont
  {Petruccione},\ and\ \citenamefont {Lloyd}}]{rebentrost2016quantum}%
  \BibitemOpen
  \bibfield  {author} {\bibinfo {author} {\bibfnamefont {P.}~\bibnamefont
  {Rebentrost}}, \bibinfo {author} {\bibfnamefont {M.}~\bibnamefont {Schuld}},
  \bibinfo {author} {\bibfnamefont {L.}~\bibnamefont {Wossnig}}, \bibinfo
  {author} {\bibfnamefont {F.}~\bibnamefont {Petruccione}}, \ and\ \bibinfo
  {author} {\bibfnamefont {S.}~\bibnamefont {Lloyd}},\ }\href@noop {}
  {\bibfield  {journal} {\bibinfo  {journal} {arXiv preprint arXiv:1612.01789}\
  } (\bibinfo {year} {2016})}\BibitemShut {NoStop}%
\bibitem [{\citenamefont {Ciliberto}\ \emph {et~al.}(2018)\citenamefont
  {Ciliberto}, \citenamefont {Herbster}, \citenamefont {Ialongo}, \citenamefont
  {Pontil}, \citenamefont {Rocchetto}, \citenamefont {Severini},\ and\
  \citenamefont {Wossnig}}]{ciliberto2018quantum}%
  \BibitemOpen
  \bibfield  {author} {\bibinfo {author} {\bibfnamefont {C.}~\bibnamefont
  {Ciliberto}}, \bibinfo {author} {\bibfnamefont {M.}~\bibnamefont {Herbster}},
  \bibinfo {author} {\bibfnamefont {A.~D.}\ \bibnamefont {Ialongo}}, \bibinfo
  {author} {\bibfnamefont {M.}~\bibnamefont {Pontil}}, \bibinfo {author}
  {\bibfnamefont {A.}~\bibnamefont {Rocchetto}}, \bibinfo {author}
  {\bibfnamefont {S.}~\bibnamefont {Severini}}, \ and\ \bibinfo {author}
  {\bibfnamefont {L.}~\bibnamefont {Wossnig}},\ }\href@noop {} {\bibfield
  {journal} {\bibinfo  {journal} {Proceedings Of The Royal Society A:
  Mathematical, Physical and Engineering Sciences}\ }\textbf {\bibinfo {volume}
  {474}},\ \bibinfo {pages} {20170551} (\bibinfo {year} {2018})}\BibitemShut
  {NoStop}%
\bibitem [{\citenamefont {Wiebe}\ \emph {et~al.}(2012)\citenamefont {Wiebe},
  \citenamefont {Braun},\ and\ \citenamefont {Lloyd}}]{wiebe2012quantum}%
  \BibitemOpen
  \bibfield  {author} {\bibinfo {author} {\bibfnamefont {N.}~\bibnamefont
  {Wiebe}}, \bibinfo {author} {\bibfnamefont {D.}~\bibnamefont {Braun}}, \ and\
  \bibinfo {author} {\bibfnamefont {S.}~\bibnamefont {Lloyd}},\ }\href@noop {}
  {\bibfield  {journal} {\bibinfo  {journal} {Physical review letters}\
  }\textbf {\bibinfo {volume} {109}},\ \bibinfo {pages} {050505} (\bibinfo
  {year} {2012})}\BibitemShut {NoStop}%
\bibitem [{\citenamefont {Childs}(2010)}]{childs2010relationship}%
  \BibitemOpen
  \bibfield  {author} {\bibinfo {author} {\bibfnamefont {A.~M.}\ \bibnamefont
  {Childs}},\ }\href@noop {} {\bibfield  {journal} {\bibinfo  {journal}
  {Communications in Mathematical Physics}\ }\textbf {\bibinfo {volume}
  {294}},\ \bibinfo {pages} {581} (\bibinfo {year} {2010})}\BibitemShut
  {NoStop}%
\bibitem [{\citenamefont {Berry}\ and\ \citenamefont
  {Childs}(2009)}]{berry2009black}%
  \BibitemOpen
  \bibfield  {author} {\bibinfo {author} {\bibfnamefont {D.~W.}\ \bibnamefont
  {Berry}}\ and\ \bibinfo {author} {\bibfnamefont {A.~M.}\ \bibnamefont
  {Childs}},\ }\href@noop {} {\bibfield  {journal} {\bibinfo  {journal} {arXiv
  preprint arXiv:0910.4157}\ } (\bibinfo {year} {2009})}\BibitemShut {NoStop}%
\bibitem [{\citenamefont {Liu}\ and\ \citenamefont
  {Zhang}(2015)}]{liu2015fast}%
  \BibitemOpen
  \bibfield  {author} {\bibinfo {author} {\bibfnamefont {Y.}~\bibnamefont
  {Liu}}\ and\ \bibinfo {author} {\bibfnamefont {S.}~\bibnamefont {Zhang}},\
  }in\ \href@noop {} {\emph {\bibinfo {booktitle} {International Workshop on
  Frontiers in Algorithmics}}}\ (\bibinfo {organization} {Springer},\ \bibinfo
  {year} {2015})\ pp.\ \bibinfo {pages} {204--216}\BibitemShut {NoStop}%
\bibitem [{\citenamefont {Hawkins}(2004)}]{hawkins2004problem}%
  \BibitemOpen
  \bibfield  {author} {\bibinfo {author} {\bibfnamefont {D.~M.}\ \bibnamefont
  {Hawkins}},\ }\href@noop {} {\bibfield  {journal} {\bibinfo  {journal}
  {Journal of chemical information and computer sciences}\ }\textbf {\bibinfo
  {volume} {44}},\ \bibinfo {pages} {1} (\bibinfo {year} {2004})}\BibitemShut
  {NoStop}%
\bibitem [{\citenamefont {Hoerl}\ and\ \citenamefont
  {Kennard}(1970)}]{hoerl1970ridge}%
  \BibitemOpen
  \bibfield  {author} {\bibinfo {author} {\bibfnamefont {A.~E.}\ \bibnamefont
  {Hoerl}}\ and\ \bibinfo {author} {\bibfnamefont {R.~W.}\ \bibnamefont
  {Kennard}},\ }\href@noop {} {\bibfield  {journal} {\bibinfo  {journal}
  {Technometrics}\ }\textbf {\bibinfo {volume} {12}},\ \bibinfo {pages} {55}
  (\bibinfo {year} {1970})}\BibitemShut {NoStop}%
\bibitem [{\citenamefont {Wossnig}\ \emph {et~al.}(2018)\citenamefont
  {Wossnig}, \citenamefont {Zhao},\ and\ \citenamefont
  {Prakash}}]{wossnig2018quantum}%
  \BibitemOpen
  \bibfield  {author} {\bibinfo {author} {\bibfnamefont {L.}~\bibnamefont
  {Wossnig}}, \bibinfo {author} {\bibfnamefont {Z.}~\bibnamefont {Zhao}}, \
  and\ \bibinfo {author} {\bibfnamefont {A.}~\bibnamefont {Prakash}},\
  }\href@noop {} {\bibfield  {journal} {\bibinfo  {journal} {Physical review
  letters}\ }\textbf {\bibinfo {volume} {120}},\ \bibinfo {pages} {050502}
  (\bibinfo {year} {2018})}\BibitemShut {NoStop}%
\bibitem [{\citenamefont {Meng}\ \emph {et~al.}(2018)\citenamefont {Meng},
  \citenamefont {Yu}, \citenamefont {Xiang},\ and\ \citenamefont
  {Zhang}}]{meng2018quantum}%
  \BibitemOpen
  \bibfield  {author} {\bibinfo {author} {\bibfnamefont {F.-X.}\ \bibnamefont
  {Meng}}, \bibinfo {author} {\bibfnamefont {X.-T.}\ \bibnamefont {Yu}},
  \bibinfo {author} {\bibfnamefont {R.-Q.}\ \bibnamefont {Xiang}}, \ and\
  \bibinfo {author} {\bibfnamefont {Z.-C.}\ \bibnamefont {Zhang}},\ }\href@noop
  {} {\bibfield  {journal} {\bibinfo  {journal} {IEEE Access}\ } (\bibinfo
  {year} {2018})}\BibitemShut {NoStop}%
\bibitem [{\citenamefont {Yu}\ \emph {et~al.}(2017)\citenamefont {Yu},
  \citenamefont {Gao},\ and\ \citenamefont {Wen}}]{yu2017quantum}%
  \BibitemOpen
  \bibfield  {author} {\bibinfo {author} {\bibfnamefont {C.-H.}\ \bibnamefont
  {Yu}}, \bibinfo {author} {\bibfnamefont {F.}~\bibnamefont {Gao}}, \ and\
  \bibinfo {author} {\bibfnamefont {Q.-Y.}\ \bibnamefont {Wen}},\ }\href@noop
  {} {\bibfield  {journal} {\bibinfo  {journal} {arXiv preprint
  arXiv:1707.09524}\ } (\bibinfo {year} {2017})}\BibitemShut {NoStop}%
\bibitem [{\citenamefont {Kitaev}(1995)}]{kitaev1995quantum}%
  \BibitemOpen
  \bibfield  {author} {\bibinfo {author} {\bibfnamefont {A.~Y.}\ \bibnamefont
  {Kitaev}},\ }\href@noop {} {\bibfield  {journal} {\bibinfo  {journal} {arXiv
  preprint quant-ph/9511026}\ } (\bibinfo {year} {1995})}\BibitemShut {NoStop}%
\bibitem [{\citenamefont {Svergun}(1992)}]{svergun1992determination}%
  \BibitemOpen
  \bibfield  {author} {\bibinfo {author} {\bibfnamefont {D.}~\bibnamefont
  {Svergun}},\ }\href@noop {} {\bibfield  {journal} {\bibinfo  {journal}
  {Journal of applied crystallography}\ }\textbf {\bibinfo {volume} {25}},\
  \bibinfo {pages} {495} (\bibinfo {year} {1992})}\BibitemShut {NoStop}%
\bibitem [{\citenamefont {Brassard}\ \emph {et~al.}(2002)\citenamefont
  {Brassard}, \citenamefont {Hoyer}, \citenamefont {Mosca},\ and\ \citenamefont
  {Tapp}}]{brassard2002quantum}%
  \BibitemOpen
  \bibfield  {author} {\bibinfo {author} {\bibfnamefont {G.}~\bibnamefont
  {Brassard}}, \bibinfo {author} {\bibfnamefont {P.}~\bibnamefont {Hoyer}},
  \bibinfo {author} {\bibfnamefont {M.}~\bibnamefont {Mosca}}, \ and\ \bibinfo
  {author} {\bibfnamefont {A.}~\bibnamefont {Tapp}},\ }\href@noop {} {\bibfield
   {journal} {\bibinfo  {journal} {Contemporary Mathematics}\ }\textbf
  {\bibinfo {volume} {305}},\ \bibinfo {pages} {53} (\bibinfo {year}
  {2002})}\BibitemShut {NoStop}%
\bibitem [{\citenamefont {Montavon}\ \emph {et~al.}(1998)\citenamefont
  {Montavon}, \citenamefont {Orr},\ and\ \citenamefont
  {M{\"u}ller}}]{montavon1998tricks}%
  \BibitemOpen
  \bibfield  {author} {\bibinfo {author} {\bibfnamefont {G.}~\bibnamefont
  {Montavon}}, \bibinfo {author} {\bibfnamefont {G.~B.}\ \bibnamefont {Orr}}, \
  and\ \bibinfo {author} {\bibfnamefont {K.-R.}\ \bibnamefont {M{\"u}ller}},\
  }\href@noop {} {\  (\bibinfo {year} {1998})}\BibitemShut {NoStop}%
\end{thebibliography}%
\bibliographystyle{apsrev4-1}
\end{document}